\documentclass[letterpaper,11pt]{article} 
\usepackage[letterpaper,hmargin=1in,vmargin=1in]{geometry} 
\usepackage{fullpage}
\usepackage[pdfstartview=FitH,colorlinks,linkcolor=blue,citecolor=blue]{hyperref} 

\usepackage[T1]{fontenc}
\usepackage{amsthm} 
\usepackage{color} 
\usepackage{times} 
\usepackage{microtype,pdfsync} 
\usepackage{amsmath,amsfonts,amssymb}
\usepackage{url} 
\usepackage{xspace,paralist}
\usepackage{lmodern}
\usepackage{float}
\floatstyle{boxed}
\restylefloat{figure}
\usepackage{multirow}

\usepackage{algpseudocode}
\usepackage{algorithm}
\usepackage{makecell}

\usepackage{xfrac}

\theoremstyle{plain} 

\newtheorem{theorem}{Theorem}[section] 
\newtheorem{lemma}[theorem]{Lemma} 
 
\newtheorem{corollary}[theorem]{Corollary}

\theoremstyle{definition}

\numberwithin{equation}{section}

\newcommand{\F}{{\mathbb F}} 
\newcommand{\Q}{\mathbb{Q}} 
\newcommand{\Z}{\mathbb{Z}}

\newcommand{\C}{\mathbb{C}} 
\newcommand{\G}{\mathbb{G}}

\newcommand{\av}{{\mathbf{a}}}
\newcommand{\bv}{{\mathbf{b}}}

\newcommand{\gv}{{\mathbf{g}}}
\newcommand{\hv}{{\mathbf{h}}}

\newcommand{\rv}{{\mathbf{r}}}
\newcommand{\sv}{{\mathbf{s}}}

\newcommand{\xv}{{\mathbf{x}}}

\newcommand{\zv}{{\mathbf{z}}}

\newcommand{\xiv}{{\boldsymbol\xi}}
\newcommand{\omegav}{{\boldsymbol\omega}}

\newcommand{\zerom}{{\boldsymbol{0}}}

\newcommand{\Bm}{{\mathbf{B}}}
\newcommand{\Cm}{{\mathbf{C}}}
\newcommand{\Dm}{{\mathbf{D}}}

\newcommand{\Hm}{{\mathbf{H}}}

\newcommand{\Mm}{{\mathbf{M}}}

\newcommand{\As}{{\mathcal{A}}}
\newcommand{\Bs}{{\mathcal{B}}}
\newcommand{\Cs}{{\mathcal{C}}}

\newcommand{\Hs}{{\mathcal{H}}}

\newcommand{\Os}{{\mathcal{O}}}
\newcommand{\Qs}{{\mathcal{Q}}}

\newcommand{\Xs}{{\mathcal{X}}}
\newcommand{\Ys}{{\mathcal{Y}}}
\newcommand{\Zs}{{\mathcal{Z}}}

\newcommand{\Tr}{{\mathrm{Tr}}}

\newcommand{\ignore}[1]{}

\begin{document}

\title{Quantum Oracle Classification - The Case of Group Structure}
\author{{\sc Mark Zhandry} \\
	Massachusetts Institute of Technology, USA \\
	{\tt mzhandry@gmail.com} }

\date{}
\maketitle

\begin{abstract} The Quantum Oracle Classification (QOC) problem is to classify a function, given only quantum black box access, into one of several classes without necessarily determining the entire function.  Generally, QOC captures a very wide range of problems in quantum query complexity.  However, relatively little is known about many of these problems.
	
In this work, we analyze the a subclass of the QOC problems where there is a group structure.  That is, suppose the range of the unknown function $A$ is a commutative group $G$, which induces a commutative group law over the entire function space.  Then we consider the case where $A$ is drawn uniformly at random from some subgroup $\As$ of the function space.  Moreover, there is a homomorpism $f$ on $\As$, and the goal is to determine $f(A)$.  This class of problems is very general, and covers several interesting cases, such as oracle evaluation\footnote{What we call oracle \emph{evaluation} has been called oracle \emph{interrogation} by~\cite{FOCS:VanDam98,EC:BonZha13}};  polynomial interpolation, evaluation, and extrapolation; and parity.  These problems are important in the study of message authentication codes in the quantum setting, and may have other applications.

We exactly characterize the quantum query complexity of every instance of QOC with group structure in terms of a particular counting problem.  That is, we provide an algorithm for this general class of problems whose success probability is determined by the solution to the counting problem, and prove its exact optimality.  Unfortunately, solving this counting problem in general is a non-trivial task, and we resort to analyzing special cases.  Our bounds unify some existing results, such as the existing oracle evaluation and parity bounds.  In the case of polynomial interpolation and evaluation, our bounds give new results for secret sharing and information theoretic message authentication codes in the quantum setting.  

\end{abstract}

\section{Introduction}

\label{sec:intro}

The quantum oracle classification (QOC) problem is to classify the function, given only quantum black box access, into one of several classes, perhaps without necessarily determining the entire function.  This very broad class of problems captures many existing quantum query problems in the literature, including all decisional problems, as well as some computational problems such as Quantum Polynomial Interpolation~\cite{KK10}, Quantum Oracle Evaluation~\cite{FOCS:VanDam98}, and others.  

In this work, we study a subset of the oracle classification problem where the oracles are promised to belong to a subspace $\As$ of the set of all functions, and the goal is to evaluate some homomorphism $f$ on the oracle.  More precisely, we consider oracles $O:\Xs\rightarrow G$, for some input domain $\Xs$ and codomain $G$, where $G$ is an abelian group.  The group structure of $G$ naturally induces a group structure on the set of all oracles, where addition of oracles is performed input-by-input.  Then we consider $\As$ to be a subgroup of the set of all oracles, and $f$ is a homomorphism on $\As$.  We even consider an ``adversarial'' setting where the algorithm is allowed to choose $f$ for itself from some set of ``allowed'' homomorphisms $\Hs$, potentially even after making queries to the oracle.

This general class of QOC problems, while excluding many decisional problems, still contains many computational problems, such as:

\paragraph{Quantum Summation (QSum).} The Parity problem asks, given an oracle $A:\Xs\rightarrow\{0,1\}$, determine the parity of $A$: $\sum_{x\in\Xs}A(x)\mod 2$.  More generally, we consider the summation problem, where the goal is to, given $A:\Xs\rightarrow G$ for an additive group $G$, determine $f(A)=\sum_{x\in\Xs}A(x)$.  It is easy to see that $f$ is a homomorphism on the function space.  Classically, this problem requires querying the entire domain to achieve any success probability better than random guessing.  In the quantum setting, this problem was initially studied by Farhi et al.~\cite{FGGS98} and Beals et al.~\cite{BBCMdW01}, who show that, in the binary case, $\lceil|\Xs|/2\rceil$ queries are sufficient to determine the parity with probability 1, and with any fewer queries it is impossible to do better than random guessing.  Thus, quantum algorithms can obtain a factor of 2 speedup relative to classical algorithms in the binary case, even when considering bounded-error algorithms.  To the best of our knowledge, this problem has not been studied in the more general setting.  It is natural to ask: does the factor of 2 speed up generalize to large outputs, or does the speedup vary based on the group?

\paragraph{Quantum Oracle Evaluation (QOE).} Given $q$ queries to a function $A:\Xs\rightarrow\Ys$, and $k$ distinct inputs $x_1,\dots,x_k$, evaluate $A$ on each of the $x_i$.  That is, output $(y_1,\dots,y_k)$ where $y_i=A(x_i)$.  This is encoded as an OCP problem by interpreting $\Ys$ as $\Z_N$, where $N=|\Ys|$.  Then the map $A\mapsto (A(x_1),\dots,A(x_k))$ is a homomorphism on the function space.  The ``adversarial'' setting corresponds to the algorithm getting to choose the inputs $x_1,\dots,x_k$, and hence the homomorphism $f$, potentially \emph{after} making the $q$ queries.

Classically, clearly with $q<k$ queries, it is impossible to do better than randomly guessing $k-q$ points.  In the quantum setting, this problem was initially studied by van Dam~\cite{FOCS:VanDam98}, who showed that in the case of binary outputs, it is possible to set $q=0.5001 k$, and succeed with overwhelming probability.  Boneh and Zhandry~\cite{EC:BonZha13} generalize van Dam's algorithm to general $N$, and give an exact matching lower bound.  They show that, for any constant $N$, it is possible to set $q$ to be a constant fraction of $k$ and succeed with overwhelming probability, though the constant goes to 1 as $N$ grows.  For $N\gg k$, they show that even for $q=k-1$, the best success probability is at most $k/N\ll 1$, showing that in this regime quantum algorithms offer little advantage over classical algorithms.

\paragraph{Quantum Polynomial Interpolation (QPI).} Given $q$ queries to a degree-$d$ polynomial $A:\F\rightarrow \F$, determine the polynomial $A$, where $\F$ is some finite field.  Using the additive group of $\F$, $\As$ forms a subgroup of the space of all functions.  Determining the polynomial is equivalent to determining its coefficients, and the map that takes a degree-$d$ polynomial $A$ to its coefficients is a homomorphism on $\As$.  Thus this is an instance of the group QOC problem.  This problem was first studied by Kane and Kutin~\cite{KK10}, who show that $q\geq (d+1)/2$ queries are required to interpolate with probability better than $1/|\F|$.  Boneh and Zhandry~\cite{EC:BonZha13} show that it is possible to interpolate with overwhelming probability when $q=d-1$.  Their techniques can also be used to show that for $q=(d+1)/2$, the success probability is at most $1/q!=1/((d+1)/2)!$.  Thus $q\geq (d+2)/2$ queries are required to interpolate with probability close to 1 for any $d\geq 2$.  However, for $(d+2)/2 < q < d-1$, this problem remained unresolved.

\paragraph{Quantum Polynomial Evaluation (QPEv).} We can also consider several variants of the polynomial interpolation problem.  In Polynomial Evaluation, we are tasked with determining the value of $A(x_i)$ for $k$ distinct points $x_1,\dots,x_k$.  Like Oracle Interrogation, it is straightforward to phrase Polynomial Evaluation as a group QOC problem.  Clearly, if $q\geq k$, this problem is easy even with only classical access: simply query $A$ on each of the $x_i$.  Therefore, the interesting case is when $q<k$.  Notice that for $d<k$, the problem of evaluating $A$ on $k$ points is equivalent to determining the polynomial $A$ itself, as the $k$ points can be interpolated to give the entire polynomial.  However, for $d\geq k$, polynomial evaluation is potentially an easier problem, which makes lower bounds harder.  

This problem relates to information-theoretic \emph{message authentication codes} (MACs) in the quantum setting.  Here, a document is ``signed'' by interpreting the document as a field element, and then applying $A$ to it.  We require that, after seeing $q$ signed documents, an adversary cannot produce an additional signed document.  In the quantum setting, we allow the adversary to even see $q$ superpositions of signed documents, and hope that the adversary cannot then produce $q+1$ classical signed documents.  We even allow the (superpositions) of documents to be chosen by the adversary, as well as the $q+1$ ``forged'' documents.  Such a MAC is called a $q$-time MAC.  When setting $A$ to be a random degree-$d$ polynomial, this corresponds to the Polynomial Evaluation problem with $k=q+1$.  Here, it is important to consider the ``adversarial'' setting, where the adversary can choose the forged documents \emph{after} making its queries.

Classically, degree $q$ polynomials suffice for $q$-time MACs, because $q$ input/output pairs of a polynomial are insufficient to determine the polynomial at any other points.  As Boneh and Zhandry~\cite{EC:BonZha13} showed, $q$ queries are sufficient to interpolate a degree $q$ polynomial entirely, so degree $q$ polynomials are \emph{not} quantum-secure $q$-time MACs.  Boneh and Zhandry show that degree $3q$ polynomials are sufficient for $q$-time MACs, but leave open exactly how high a degree is necessary.  In particular, for 1-time MACs, whereas classically degree 1 is sufficient, quantumly degree 1 is insufficient, but degree 3 is sufficient.  Whether degree 2 polynomials are 1-time MACs is currently unknown.  Answering this question is of practical importance, as the size of the key (that is, the descrption of the polynomial) and evaluation time grow with the degree.

\paragraph{Quantum Polynomial Extrapolation (QPEx).} A final variant of the polynomial interpolation problem we consider is Polynomial Extrapolation.  Here, the there is a particular point, say 0.  The adversary is trying to determine $A(0)$, and is given oracle access to the polynomial everywhere else.  Again, it is natural to encode this as a group QOC problem.

This problem relates to a quantum analog of Shamir secret sharing~\cite{Shamir79}.  In secret sharing, a secret $s\in\F$ is split into $k$ shares $sh_1,\dots,sh_k\in\F$, which are distributed to each of $k$ parties.  The requirement is that any subset of $t+1$ parties can reconstruct the secret, but any subset of $t$ or fewer parties cannot.  Shamir builds such a scheme by setting $A$ to be a random degree-$t$ polynomial conditioned on $A(0)=s$, and setting $sh_i=A(i)$.  Then from any $t+1$ shares, $A$ can be classically interpolated, but from any $t$ shares, it cannot.  Thus there is a tight threshold, above which the secret can be reconstructed, and below which the secret remains hidden.

Damg{\aa}rd et al.~\cite{DFNS14} explore superposition attacks on such a scheme, where an adversary can obtain superpositions of shares subsets of up to $t$ users, which corresponds to making $t$ quantum queries to $A$.  They show that degree-$t$ polynomials are insufficient for such an attack model, a result that also follows from the fact that degree-$t$ polynomials can be interpolated with $t$ queries.  Therefore, degree $d\geq t+1$ is required.  However, such a polynomial now requires at $d+1\geq t+2$ honest users to reconstruct. Thus, there is no tight threshold between classical honest reconstruction and quantum dishonest reconstruction.

It is natural to ask, then, what happens if the shares are themselves superpositions.  Is it now possible to have a tight threshold as in the classical case?  The hope would be that, for any $t$, there is some degree $d$ such that any $t$ quantum queries is insufficient to reconstruct the secret, but $t+1$ quantum queries are sufficient.  This corresponds exactly to asking if the Polynomial Extrapolation problem has a tight threshold above which the problem can be solved, but below which it cannot.

\subsection{Our Results}

We give \emph{exact} bounds on the maximal success probability for a quantum oracle adversary for any instance of the group QOC problem.  That is, we provide an algorithm achieving the maximal success probability, and show that this probability is exactly optimal.  Our expression for this probability is given as the solution to a certain counting problem that depends on the particular group QOC instance.  Our general bounds show several interesting things:
\begin{itemize}
	\item Our algorithm is a \emph{parallel} algorithm, meaning all the queries are made simultaneously in parallel, rather that sequentially.  Yet our lower bound is for arbitrary \emph{sequential} algorithms.  Thus, for any instance of the group QOC problem, the best possible algorithm is a parallel algorithm.  This is interesting because many quantum speedups require sequential queries;  For example Grover's algorithm is inherently and necessarily sequential~\cite{Zalka99}.  Such speedups are not applicable to group QOC instances.
	\item In the case where the adversary can choose the homomorphism $f$ from some class $\Hs$, our algorithm chooses $f$ \emph{before} making any queries, whereas our lower bound is for algorithms that may choose $f$ adversarially, \emph{after} making all queries.  Thus, our results show that for any instance of the group QOC problem, there is no advantage to choosing $f$ at the end.  This greatly generalizes a result of Boneh and Zhandry~\cite{EC:BonZha13}, which show that this holds for the Oracle Evaluation problem.
	\item In the case where the homomorphism attains one of two possible values, our analysis shows that below a certain threshold number of queries it is impossible to do better than randomly guessing the output, and above the threshold the correct output can be obtained with probability 1.  Thus there is a tight threshold between below which quantum queries are useless and above which there is an exact quantum algorithm.  More generally, the maximum success probability is always an integer multiple of the success probability obtained by random guessing.
	\item While we state our bounds as average case results, they apply equally well to the \emph{worst} case.  In particular, our average case lower bound immediately gives a worst case lower bound, as any algorithm that has success probability at least $p$ in the worst case also has success probability at least $p$ on average.  Moreover, our algorithm is a worst-case algorithm, in that its success probability is independent of the oracle it is given.  This is because the QOC problem with group structure admits a query-preserving \emph{randomized self reduction}, meaning that any worst case oracle can be turned into a random oracle, and a solution for the random oracle gives a solution for the worst case oracle.
\end{itemize}

Next, we turn to analyzing the exact bound itself.  Unfortunately, in general determining simple expressions for the solutions to our counting problems is difficult.  Instead, we analyze several special cases, yielding the following results:
\begin{itemize}
	\item In the case of Parity (a special case of Summation) and Oracle Evaluation, we re-prove the existing results~\cite{FGGS98,BBCMdW01,EC:BonZha13} using our new characterization, thus unifying these results.  
	\item In the case of Summation, we generalize the known bounds to computing the sum of outputs, where the range of $A$ is an arbitrary group $G$.  We show that, as $G$ increases, the quantum speedup diminishes.  In particular, the quantum speedup is a factor of  $1+O(1/|G|)$, even if we allow bounded error algorithms.  For very large $G$, the quantum speedup becomes negligible.
	\item For Polynomial Interpolation, we show that $q$ queries to a degree-$d$ polynomial for $q\leq d/2$ results in a success probability at most $1/|\F|$, where $\F$ is the field.  This reproves a result of Kane and Kutin~\cite{KK10}.  We also show that for $q\geq d/2+1$, the success probability approaches 1 for $|\F|\gg d!$.  Thus, for constant \emph{even} $d$, there is a sharp threshold at $d/2$, at or below which interpolation is impossible, and above which interpolation is possible with overwhelming probability.  Interestingly, we show that for $q=d/2+1/2$, the success probability is $\approx 1/q!$.  For constant \emph{odd} $d$, this is constant, so there is \emph{no sharp threshold} in this case.  Taking $q$ to be large, however, makes this negligible and gives a sharp threshold at $d/2+1/2$.
	\item For Polynomial Evaluation, the above result for Interpolation shows that $q\geq d/2+1$ suffices to reconstruct the entire polynomial (in the donstant $d$, large $|\F|$ case), and therefore evaluate it on any number of points.  Moreover, $q=d/2+1/2$ suffices for complete reconstruction with constant probability (independent of $|\F|$ but dependent on $d$).  Thus Polynomial Evaluation is \emph{easy} in this regime.  Therefore, degree $d=2q-1$ polynomials are \emph{not} secure $q$-time MACs.  On the flip side, we show that for $q\leq d/2$, the success probability in the Polynomial Evaluation problem vanishes for large $\F$, meaning degree $d=2q$ polynomials \emph{are} $q$-time MACs, which is the best possible.  This improves on the best prior result of Boneh and Zhandry~\cite{EC:BonZha13}, who show that degree $d=3q$ are $q$-time MACs.  In particular, this shows that degree-2 polynomials are 1-time MACs, resolving an open question of Boneh and Zhandry.
	\item For Polynomial Extrapolation, our Interpolation result again shows that $q\geq d/2+1$ suffices to reconstruct the polynomial and evaluate at 0\footnote{Technically, we cannot run the Interpolation algorithm as is, since the Interpolation problem is allowed to evaluate at 0.  However, 0 is just a single point and a negligible fraction of the domain.  Eliminating 0 from the queryable domain therefore results in a negligible decrease in success probability}.  We show that for $q\leq d/2$, it is impossible to do better than random guessing, again reproving a result of Kane and Kutin~\cite{KK10}.  Therefore, for even $d$, we have a tight threshold at $d/2$ between being unable to reconstruct the secret and reconstructing the secret.  Put another way, to have threshold $q$, it suffices to set $d=2q$.  This shows that quantum secret sharing may be possible with a tight threshold, though one would need to show that the shares can be created in such a way as to allow for reconstruction.  
	
	One may ask if $d=2q-1$ suffices for secret sharing with threshold $q$, as using lower degrees could improve efficiency.  For odd $d$, we show for some field sizes that the secret can be reconstructed with probability $\approx 1/q$ for $q=d/2+1/2$, and we conjecture this holds for all field sizes.  Since $q$ is always a polynomial, this cannot be made non-negligible by making $q$ large.  Therefore, setting $d=2q-1$ cannot result in a tight threshold, and instead $d=2q$ is required.
\end{itemize}

\subsection{Our Techniques}

In order to establish our characterization of quantum query complexity in terms of a counting problem, we need to provide an algorithm and a matching lower bound.  

\medskip

For our lower bound, many traditional quantum lower bound techniques such as the adversary method and its generalizations and the polynomial method cannot give the results we need.  One reason is because these techniques inherently give asymptotic query lower bounds.  However, we see that for the particular instances of the group QOC problem studied here, the the quantum speed up over classical algorithms is a factor of 2 or less, and we are therefore interested in the \emph{exact} number of queries required to solve the problem.  An asymptotic bound here is relatively meaningless.

Instead, our lower bound is based on the Rank method of Boneh and Zhandry~\cite{EC:BonZha13}, which does give exact lower bounds.  There, they show that the maximal success probability of any quantum query algorithm is bounded by the success probability of random guessing, times a quantity called the ``rank'' of the algorithm.  The rank of a quantum algorithm is just the dimension of the space spanned by possible output states of the algorithm on different oracles.  Boneh and Zhandry also give a bound on the rank for any oracle algorithm, which equal to the solution to a very simple counting problem: the number of functions that differ from 0 on at most $q$ points.  For the case of Polynomial Interrogation, the resulting bound turns out to be exactly optimal.

Unfortunately, Boneh and Zhandry's general bound on the rank is too weak for many interesting problems.  One problem is that the bound on the rank applies in the case where the oracle can be \emph{any} possible function.  However, in many settings, the oracle is restricted: for example in polynomial interpolation, the function is a degree-$d$ polynomial.  Therefore, while the exiting bounds on rank for general oracles give some bound on the success probability in the restricted setting, it is unlikely to be optimal.

Another problem is that, even if a new tighter rank bound is obtained, it will often give a meaningless result for the maximal success probability.  In particular, the rank method appears to be most useful when the goal is to determine the the oracle completely, and typically gives meaningless bounds if the goal is to determine some partial information about the oracle.  This is because every possible oracle the adversary may be given could potentially contribute to the rank.  Therefore, even if the rank is much smaller than the number of possible oracles, it will still be quite large.  If the goal is to determine only a small amount of information about the oracle (meaning the random guessing probability is not too small), the rank method will give meaningless results.

To make this concrete, consider the case of Polynomial Extrapolation.  With $q=d/2$ queries, an analysis of our algorithm for the \emph{Polynomial Interpolation} problem shows that the rank of a quantum algorithm making $q$ queries to a degree $d=2q$ polynomial can be as high as roughly $|\F|^{2q}/q!$.  If we try to apply the rank method directly to Polynomial Extrapolation, which has a random-guessing success probability of $1/|\F|$, we get an upper bound of $|\F|^{2q-1}/q!\gg 1$ for the success probability.  

To overcome these difficulties, we extend the rank analysis of Boneh and Zhandry to handle density matrices.  That is, for every possible value $v$ in the image of of the homomorphism $f$, we consider the density matrix of final states obtained when given oracle access to a random oracle $A$ conditioned on $f(A)=v$.  Our goal is to bound the ability to distinguish these density matrices for different $v$.  In general, this appears to be a very difficult problem.  However, we use the group structure of the problem to reduce analyzing the distinguishing probability for these density matrices to analyzing the distinguishing probability of certain pure states.  We can then bound the dimension of the space spanned by these pure states (which can be thought of as the effective rank of the algorithm) as the result of a particular counting problem.  Then applying the Rank method of Boneh and Zhandry gives the final result.

\medskip

For our algorithm, we set up a superposition of parallel queries, where each vector of inputs corresponds to one item in the counting problem.  We then show that, after making all of the queries in parallel, the resulting states are in a sense maximally distinguishable.  We obtain a success probability that exactly equals the bound obtained above.

\medskip

Finally, to actually apply our general group QOC result, we need to solve the obtained counting problem for the particular instances of interest.  This turns out to be a non-trivial task, and the problem has very different structures depending on which problem we are solving.  Nonetheless, we give solutions to the problem in the case of QSum, QUE, QPI,QPEv, and QPEx, yielding our results outlined above.

\subsection{Independent and Concurrent Work}

In a very recent concurrent and independent work, Childs et al.~\cite{CvDHS15} also give exact bounds for the Polynomial Interpolation problem discussed above.  In particular, they show that the optimal success probability relates to a particular counting problem.  The counting problem they obtain is exactly the same 
 problem we obtain when applying our new technique to the Polynomial Interpolation problem.   Thus their main result is a special case of our analysis.  They do not analyze the other related problems, such as Polynomial Evaluation or Polynomial Extrapolation.  However, they go a step further in their analysis of the Polynomial Interpolation problem, and give efficient quantum algorithms with essentially optimal success probability.  In contrast, our quantum algorithm, due to its generality in being able to handle \emph{any} oracle classification problem with group structure, is not efficient.

\section{Preliminaries}

\label{sec:prelim}

We will assume familiarity with basic group theory.  Fix a commutative group $G$.  A \emph{character} $\xi:G\rightarrow \C^\times$ of $G$ is any homomorphism from $G$ to the multiplicative group of non-zero complex numbers $\C^\times$ --- that is, $\xi(g+h)=\xi(g)\xi(h)$.  If $G$ has order $n$, it is easy to see that $\xi(g)$ must be an $n$th root of unity for every $g\in G$.  The characters of $G$ form a commutative group with the multiplicative group law $(\xi\cdot\xi')(g)=\xi(g)\xi'(g)$.  Call this group $\hat{G}$.  It is known that $\hat{G}$ is isomorphic to $G$.  For any character $\xi$, it is straightforward to show that $\sum_{g\in G}\xi(g)$ is zero, unless $\xi$ is the trivial homomorphism that is identically 1, in which case the sum is $|G|$. 

Given an integer $N$, let $\omega_N=e^{i 2\pi/N}$, which is a primitive $N$th-root of unity.

\paragraph{Quantum query model.}Let $\Xs$ be a set, $G$ be a finite commutative group (written additively), and consider a function $A:\Xs\rightarrow G$.  The usual model for quantum queries to $A$ is the \emph{controlled-add} model.   Here, a query to $A$ is the unitary operation defined by \[|x,g,z\rangle\rightarrow |x,g+A(x),z\rangle\]

Alternatively, one can consider the \emph{phase} model for quantum queries.  A \emph{phase} query to $A$ is the unitary operation defined by 
\[|x,\xi,z\rangle\rightarrow \xi(A(x))|x,\xi,z\rangle\]

where $\xi$ is a character of $\G$.  The phase model and controlled-add model are actually equivalent, by conjugating with the quantum group Fourier transform unitary, defined by
\[|x,g,z\rangle\rightarrow \frac{1}{\sqrt{|G|}}\sum_{\xi\in\hat{G}}\xi(g)|x,\xi,z\rangle\]

For this work, we will mainly use the phase query model, as it will make our calculations easier; our results also apply equally well to the controlled-add model.

\section{Oracle Classification With Group Structure}

\label{sec:OIPdef}

Let $G$ be a commutative group (written additively) and $\Xs$ some input space.  Then the set $\Os$ of functions $O:\Xs\rightarrow G$ inherits a group law from $G$: $(O+O')(x)=O(x)+O'(x)$.  Consider a subgroup $\As$ of $\Os$, and consider the uniform distribution on this set.  Consider a homomorphism $f$ from $\As$ to some other arbitrary group.  

\paragraph{Group Quantum Oracle Classification (Group QOC) Problem.} The oracle classification problem is to, given $q$ quantum oracle queries to a function $A$ drawn at random from $\As$, determine $f(A)$.  Again, the goal is to devise an algorithm that maximizes the success probability.  

\paragraph{Special Case: Group Quantum Oracle Identification (Group QOI).} The oracle identification problem is to, given $q$ quantum oracle queries to a function $A$ drawn at random from $\As$, determine $A$ completely.  The goal is to devise an algorithm that maximizes the success probability.  Notice that by setting $f$ to be the identity on $\As$, the oracle identification problem is a special case of the oracle classification problem.

\paragraph{Generalization to Adversarially-chosen $f$.}  A generalization of the Group QOC problem allows the adversary itself to choose the homomorphism $f$.  That is, the adversary makes $q$ quantum queries to $A$, and then outputs a homomorphism $f$ together with $f(A)$.  Of course, this problem is trivial if we place no restrictions on $f$: the adversary can simply output the trivial homomorphism for $f$, in which case $f(A)$ is zero.  However, we will require that $f$ comes from some restricted set $\Hs$ of subgroups.  

It may seem that allowing the adversary to choose $f$ will give it extra power: based on the oracle queries made, the adversary can choose $f$ to increase his chances of outputting a good value $f(A)$.  We argue, using a generalization of a technique of Boneh and Zhandry~\cite{EC:BonZha13}, that the adversary might as well just choose a single $f$ up front and always output that $f$.

\begin{lemma}\label{lem:adversarial} Let $\Hs$ be a set of homomorphisms of $\As$.  Let $\Qs$ be an algorithm making $q$ quantum queries to an oracle $A$ drawn from $\As$.  Suppose with probability $\epsilon$, $\Qs$ outputs a homomorphism $f$ and the correct value $f(A)$.  Then there is a fixed homomorphism $f$ and quantum algorithm $\Qs'$ such that $\Qs'$ outputs $f(A)$ with probability at least $\epsilon$.\end{lemma}

\begin{proof}First let $\Qs_0$ be the following modification to $\Qs$: $\Qs_0$, given quantum oracle access to a function $A\in\As$, chooses a random $A_0\in\As$ at random and simulates $\Qs$ with oracle $A'=A+A_0$.  Answering each of the oracle queries of $\Qs$ requires only a single oracle query to $A$.  When $\Qs$ outputs $f,v'$, $\Qs_0$ outputs $f,v=v'-f(A_0)$.
	
We make two observations:
\begin{itemize}
	\item If $\Qs$ succeeds, then $v'=f(A')=f(A+A_0)=f(A)+f(A_0)$.  Thus $v=v'-f(A_0)=f(A)$, and thus $\Qs_0$ succeeds.  Therefore, $\Qs_0$ succeeds with probability $\epsilon$.  Moreover, this probability is independent of the oracle $f$.
	\item The oracle seen by $\Qs$ is independent of the oracle $A$ seen by $\Qs_0$.  In particular, this means that $\Bs$ is independent of $A$.
\end{itemize}	

Thus we have modified our algorithm so that $f$ is independent of $A$ without hurting the success probability.  The new $\Qs_0$ is no longer a pure algorithm since it flips coins (to choose $A_0$), but it can be easily made pure by using a uniform superposition of all coins.  Let $p_f$ the the probability of obtaining $f$, and let $\epsilon_f$ be the probability of success, conditioned on obtaining $f$.  Then $\epsilon=\sum_{f\in\Hs}p_f\epsilon_f$.  In particular, there exists a $f^*$ such that $\epsilon_{f^*}\geq\epsilon$.  Fix this $f^*$.  We now devise an algorithm that always outputs $f^*$, and succeeds with probability $\epsilon_{f^*}$.  Thus is sufficient for proving the lemma.

Let $|\psi_0\rangle$ be the initial state of $\Qs_0$ before the first oracle query.  Let $|\psi_f\rangle$ be the following state.  Run $\Qs_0$ on the function $A=0$ until just before the final measurement.  Measure the subspace register, obtaining $f$ with probability $p_f$ (since the probability of obtaining $f$ is independent of $A$).  Un-compute the entire procedure; the new ``initial'' state obtained is $|\psi_f\rangle$.  Notice that if we re-compute $\Qs_0$ using oracle $f=0$ on this initial state, we will always obtain the same $f$.  Moreover, the success probability will be exactly $\epsilon_f$ since the success probability is independent of $A$, and instead only depends on $f$.

We now describe the algorithm $\Qs'$.  $\Qs'$ runs $\Qs_0$, except that is uses the initial state $|\psi_{f^*}\rangle$ instead of $|\psi_0\rangle$.  Since $f$ is independent of the oracle $A$, $\Qs'$ will always output $f^*$ (which would be obtained if the oracle was $A=0$).  Moreover, $\Qs'$ has success probability $\epsilon_{f^*}$.  
\end{proof}

Given Lemma~\ref{lem:adversarial}, it suffices to consider the case where there is just a single fixed homomorphism $f$.

\section{Our Main Theorem}

\label{sec:theorem}
Let $\Os$ be the group of all functions from $\Xs$ to $G$, and let $\As$ be a subgroup.  Let $f$ be a homomorphism from $\As$ into some other group.  Let $\Bs$ be the kernel of the homomorphism, and $\Cs$ be the quotient group $\As/\Bs$, interpreted as a subgroup of $\As$.  Then determining $f(A)$ for an element $A\in\As$ is equivalent to determining which coset of $\Bs$ $A$ belongs to, which amounts to finding a $C\in \Cs$ such that $A-C\in\Bs$.

Consider a set of inputs $\xv=\{x_1,\dots,x_q\}\in\Xs^q$ and a vector $\xiv=\{\xi_1,\dots,\xi_q\}\in\hat{G}^q$.  For $O\in\Os$, let $O(\xv)=(O(x_1),\dots,O(x_q))\in G^q$ be the result of applying $O$ on each of the $x_i$.  For $\gv\in G^q$, let $\xiv(\gv)=\prod_{i=1}^q \xi_i(g_i)$.

Define $e_{\xv,\xiv}:\As\rightarrow\C^\times$ as $e_{\xv,\xiv}(A)=\xiv(A(\xv))=\prod_{i=1}^q \xi_i(A(x_i))$.  Notice that $e_{\xv,\xiv}$ is a character of $\As$: that is, a homomorphism from $\As$ into $\C^\times$.  Let $E_{\As,q}$ be the space of all functions $e_{\xv,\xiv}$.  Notice that there may be collisions in that $e_{\xv,\xiv}=e_{\xv',\xiv'}$ as functions while $(\xv,\xiv)\neq (\xv',\xiv')$ (for example, permuting the values in both $\xv$ and $\xiv$ by the same permutation yields the same function).  We will consider such $e_{\xv,\xiv}$, $e_{\xv',\xiv'}$ to be the same function.  In particular, this means that $|E_{\As,q}|$ will be smaller than the set of pairs $(\xv,\xiv)$.

The subgroup $\Bs$ induced by the homomorphism $f$ induces an equivalence relation $\equiv_f$ on $E_{\As,q}$, where $e\equiv_f e'$ if $e$ and $e'$ are identical on $\Bs$.  For a function $e\in E_{\As,q}$, let $E_{\As,q,f,e}=\{e'\in E_{\As,q}:e\equiv_f e'\}$ be the equivalence class induced by $f$ that $e$ belongs to.  Let $E_{\As,q,f}=\{E_{\As,q,f,e}\}$ be the set of equivalence classes.

\begin{theorem}\label{thm:main} The maximum success probability of any quantum algorithm at solving the quantum oracle classification problem for $(\As,f)$ given $q$ quantum queries is exactly \[P_{\As,q,f}\equiv\frac{1}{|\Cs|}\max_{e\in E_{\As,q}}|E_{\As,q,f,e}|=\frac{|\Bs|}{|\As|}\max_{e\in E_{\As,q}}|E_{\As,q,f,e}|\]
\end{theorem}

Notice that $e_{\xv,\xiv}$ is indifferent to simultaneously permuting the coordinates of $\xv$ and $\xiv$.  Also note that  $e_{(x_1,\dots,x_{q-1},x_{q-1}),(\xi_1,\dots,\xi_{q-1},\xi_{q})}=e_{(x_1,\dots,x_{q-1},x_{q-1}),(\xi_1,\dots,\xi_{q-1}\cdot\xi_{q},1)}$ where $1$ is the character that maps all of $G$ to 1.  Therefore, when enumerating the $e_{\xv,\xiv}$, is suffices to consider $\xv$ whose elements are distinct, and in sorted order (according to some arbitrary ordering on $\Xs$).  Therefore, we will consider the vector $\xv$ as a \emph{set} of $q$ elements in $\Xs$.

\medskip

Before proving the theorem, we give an alternative formulation that is significantly easier to use.  This formulation loses some generality, but covers all of the cases considered in this work.  Decompose $G$ as $G=\Z_{N_1}\times\cdots\times \Z_{N_k}$.  Component-wise multiplication yields a ring structure on $G$.  Then the space $\Os$ of all oracles actually forms a $G$-module.  Suppose $\As,\Bs,\Cs$ are now actually submodules of $\Os$, and moreover suppose they are free submodules (that is, they have a basis).

It is then straightforward to show that a character $\xi$ on $G$ is a map of the following form: \[g=(g_1,\dots,g_k)\mapsto \omega_{N_1}^{r_1 g_1}\omega_{N_2}^{r_2 g_2}\cdots \omega_{N_k}^{r_k g_k}=\omegav^{r g}\]
where $\omegav=(\omega_{N_1},\dots,\omega_{N_k})$ is the vector of primitive $N_k$th roots of unity, $r=(r_1,\dots,r_k)\in G$ is some element in $G$. $r$ then completely describes the character $\xi$.  Now, using the assumption that $\Bs$ has basis $B_1,\dots,B_s$ and $\Cs$ has basis $C_1,\dots,C_t$, any element $A$ in $\As$ can be written as $A=\sum_{\ell=1}^s\beta_\ell B_\ell+\sum_{m=1}^t\gamma_m C_m$.  We see that

\begin{align*}e_{\xv,\xiv}(\sum_{\ell=1}^s\beta_\ell B_\ell+\sum_{m=1}^t\gamma_m C_m)&=\left(\prod_{i=1}^q \prod_{\ell=1}^s\omegav^{r_i B_\ell(x_i) \beta_\ell}\right)\left(\prod_{i=1}^q \prod_{m=1}^t\omegav^{r_i C_m(x_i) \gamma_m}\right)\\&=\left(\prod_{\ell=1}^s\omegav^{\langle \rv,B_\ell(\xv)\rangle\beta_\ell}\right)\left(\prod_{m=1}^t\omegav^{\langle \rv,C_m(\xv)\rangle\gamma_m}\right)\end{align*}

where $\rv=(r_1,\dots,r_q)$ is the vector of $r_i$'s that define the $\xi_i$'s.  The quantities $h_\ell=\langle \rv,B_\ell(\xv)\rangle$ and $z_m=\langle \rv,C_m(\xv)\rangle$ as $\ell,m$ vary characterize the action of $e_{\xv,\xiv}$ on $\As$, in that there is a bijective correspondence between the vector pairs $\hv=(h_1,\dots,h_s),\zv=(z_1,\dots,z_t)$ and the $e_{\xv,\xiv}$.  Moreover, $\hv$ characterizes the action on $\Bs$.  Therefore, if $e$ has vectors $\hv,\zv$ and $e'$ has vectors $\hv',\zv'$, the $e\equiv_f e'$ if and only if $\hv=\hv'$.

Lastly, notice that we can write

\[\hv=\left(\begin{array}{cccc}
B_1(x_1)& B_1(x_2)& \cdots & B_1(x_q)\\
B_2(x_1)& B_2(x_2)& \cdots & B_2(x_q)\\
\vdots & \vdots & \ddots & \vdots \\
B_r(x_1)& B_r(x_2)& \cdots & B_r(x_q)
\end{array}\right)\cdot \rv=\Bm(\xv)\cdot \rv\]

We can also conclude that $\zv=\Cm(\xv)\cdot\rv$, where $\Cm(\xv)$ is defined analogously to $\Bm(\xv)$.  This gives rise to the following corollary:

\begin{corollary}\label{cor:main}Given module $\As$, submodule $\Bs$, and quotient module $\Cs$, let $\Bs$ be spanned by functions $B_1,\dots,B_r$ and $\Cs$ be spanned by $C_1,\dots,C_s$.  Let \[\Bm(\xv)=\left(\begin{array}{cccc}
	B_1(x_1)& B_1(x_2)& \cdots & B_1(x_q)\\
	B_2(x_1)& B_2(x_2)& \cdots & B_2(x_q)\\
	\vdots & \vdots & \ddots & \vdots \\
	B_r(x_1)& B_r(x_2)& \cdots & B_r(x_q)
	\end{array}\right)\;\;\;\;\;\;\;\;\;\;\Cm(\xv)=\left(\begin{array}{cccc}
	C_1(x_1)& C_1(x_2)& \cdots & C_1(x_q)\\
	C_2(x_1)& C_2(x_2)& \cdots & C_2(x_q)\\
	\vdots & \vdots & \ddots & \vdots \\
	C_s(x_1)& C_s(x_2)& \cdots & C_s(x_q)
	\end{array}\right)\]
	
	Then the maximum success probability of any algorithm for the QOC problem on $\As,f$ given $q$ queries is exactly  \[\frac{1}{|\Cs|}\max_{\hv}|\{\zv:\exists \xv,\rv,\; \zv=\Cm(\xv),\; \Bm(\xv)\cdot\rv=\hv\}|\]
\end{corollary}

Therefore, understanding the maximal success probability amounts to understanding the matrices $\Bm(\xv),\Cm(\xv)$ and the spaces spanned by them.
\medskip

In Section~\ref{sec:apps}, we will the version of Theorem~\ref{thm:main} given in Corollary~\ref{cor:main} for several applications.  Next, we prove Theorem~\ref{thm:main}.

\subsection{The Algorithm}

Fix a function $e^*\in E_{\As,q}$.  We will show how to build a quantum query algorithm achieving success probability $\frac{|E_{\As,q,f,e^*}|}{|\Cs|}$.  By choosing the optimal $e^*$, we obtain an algorithm with success probability $P_{\As,q,f}$.  Let $S_{\xv,\xiv}$ denote the number of $(\xv',\xiv')$ pairs such that $e_{\xv',\xiv'}=e_{\xv,\xiv}$.  

First, the algorithm will construct the superposition \[\frac{1}{|E_{\As,q,f,e^*}|}\sum_{\xv,\xiv:\;e_{\xv,\xiv}\;\equiv_f\; e^*}\frac{1}{\sqrt{S_{\xv,\xiv}}}|\xv,\xiv\rangle\]

It is straightforward that this superposition is properly normalized.

Next, make $q$ queries in parallel on the superposition.  For an oracle $A=B+C$ where $B\in\Bs,C\in\Cs$, the resulting state is:
\[\frac{1}{\sqrt{|E_{\As,q,f,e^*}|}}\sum_{\xv,\xiv:\;e_{\xv,\xiv}\;\equiv_f \;e^*}\frac{\xiv(B(\xv))\xiv(C(\xv))|\xv,\xiv\rangle}{\sqrt{S_{\xv,\xiv}}}=\frac{e^*(B)}{\sqrt{|E_{\As,q,f,e^*}}}\sum_{\xv,\xiv:\;e_{\xv,\xiv}\;\equiv_f \;e^*}\frac{e_{\xv,\xiv}(C)|\xv,\xiv\rangle}{\sqrt{S_{\xv,\xiv}}}\]

Notice that the state is independent of $B$ except for an overall phase factor $e^*(B)$.  Denote this state as $e^*(B)|\psi_C\rangle$.  Therefore, if we trace out the oracle $B$, the phase disappears and we get a pure state $|\psi_C\rangle$.  Therefore, our algorithm needs to distinguish the states $|\psi_C\rangle$.

Let $T$ be the matrix whose column vectors are the $|\psi_C\rangle$.  Now we examine the entries of $U=T^\dagger\cdot T$, which consist of the inner products $\langle \psi_{C'} | \psi_{C}\rangle$.  It is straightforward to show that 
\[U_{C',C}=\langle \psi_{C'} | \psi_{C}\rangle=\frac{1}{|E_{\As,q,f,e^*}|}\sum_{e\in E_{\As,q,f,e^*}}e(C-C')\]

Assume for the moment that $U$ is full rank.  We will measure in the basis $R=(R_C)_{c\in\Cs}=T\cdot U^{-1/2}$, which can be verified to be a basis.  The probability that we measure $C$ is $|\langle R_C,\psi_C\rangle|^2$.  Averaging over all $C$, the probability that we guess $C$ is
\begin{equation}\frac{1}{|\Cs|}\sum_C |\langle R_C,\psi_C\rangle|^2=\frac{1}{|\Cs|}\Tr^2(R^\dagger\cdot T)=\frac{1}{|\Cs|}\Tr^2(U^{1/2})\label{eq:1}\end{equation}
where $\Tr^2$ is the sum of squares of diagonal entries.  

Now, if $U$ is not full rank, we can use a pseudoinverse in constructing $R$, and it is straightforward to show that Equation~\ref{eq:1} still holds.

Next, we claim that $U^2_{C',C}=\frac{|\Cs|}{|E_{\As,q,f,e^*}|}U_{C',C}$.  Indeed,
\begin{align*}U^2_{C',C}&= \sum_{C''}U_{C',C''}U_{C'',C} = \frac{1}{|E_{\As,q,f,e^*}|^2}\sum_{e,e'\in E_{\As,q,f,e^*},C''}e(C-C'')e'(C''-C')\\
&=\frac{1}{|E_{\As,q,f,e^*}|^2}\sum_{e,e'\in E_{\As,q,f,e^*}}e(C)e'(-C')\sum_{C''\in\Cs}(e'/e)(C'')
\end{align*}

Next, $e'/e$, when restricted to $\Cs$, is a character of $\Cs$, and if $e\neq e'$ when restricted to $\Cs$, then this character is not identically 1.  This means that if $e\neq e'$ on $\Cs$, then the sum $\sum_{C''}(e'/e)(C'')$ will vanish.  Otherwise, if they are equal on $\Cs$, the sum will become $|\Cs|$.  Therefore, $U^2_{C',C}$ becomes
\[U^2_{C',C}=\frac{|\Cs|}{|E_{\As,q,f,e^*}|^2}\sum_{e\in E_{\As,q,f,e^*}}e(C-C')=\frac{|\Cs|}{|E_{\As,q,f,e^*}|}U_{C',C}\]
as desired.  

Now we take the square root of both sides, obtaining $U^{1/2}_{C',C}=\frac{\sqrt{|E_{\As,q,f,e^*}|}}{\sqrt{|\Cs|}}U_{C',C}$.  Notice that the diagonal elements of this quantity are all $\frac{\sqrt{|E_{\As,q,f,e^*}}}{\sqrt{|C|}}$ since $U_{C,C}=1$ due to normalization.  Squaring and summing over all $C$, we obtain that $\Tr^2(U^{1/2})=|E_{\As,q,f,e^*}|$.  This gives the desired success probability of $|E_{\As,q,f,e^*}|/|\Cs|$.

\subsection{The Lower Bound}

Now we prove our lower bound (that is, our upper bound on the success probability).  At a high level, our result will use the Rank Method of Boneh and Zhandry~\cite{EC:BonZha13}.  However, the rank of an algorithm in our case is too high: the rank grows with the number of oracles ($|\As|$), but we are trying to identify an oracle in a potentially much smaller set ($|\Cs|$).  Therefore, we will have to be careful in our application of this method.

Consider a general algorithm $\Qs$, and suppose that $\Qs$ has probability $\epsilon$ in solving the oracle classification problem for $\As,f$.  Let $\Bs,\Cs$ be the subgroups of $\As$ induced by $f$.  Let $A\in \As$ be the oracle seen by $\Qs$, and write $A=B+C$ for oracles $B\in\Bs,C\in\Cs$.  The goal of $\Qs$ is to determine $C$.  To that end, we will consider the final density matrices of $\Qs$, denoted $\rho_C$, obtained by fixing $C$, but letting $B$ vary.  Our goal is to bound the the distinguishability of the density matrices $\rho_C$.  Unfortunately, beyond distinguishing two density matrices, we do no know of any general solution that, given a set of density matrices, determines how distinguishable they are.  Instead, we will have to use particular properties of the $\rho_C$ to argue indistinguishability.

It is straightforward to show that the state $|\psi_A\rangle$ obtained by running $\Qs$ on oracle $A$, but stopping just before the final measurement, can be written as \[|\psi_A\rangle=\sum_{\xv\in\Xs^q,\xiv\in\hat{G}^q,z\in\Zs}U_{\xv,\xiv,z}\xiv(A(\xv))|z\rangle\]
Here, the quantities $U_{\xv,\xiv,z}$ are determined solely by the algorithm $\Qs$ and independent of $A$, and $\Zs$ is some set: $z\in\Zs$ encodes the output $C'\in\Cs$, as well as some auxiliary information that is discarded.

As $\xiv(A(\xv))=e_{\xv,\xiv}(A)$ and $A=B+C$, we can write this as 

\[|\psi_A\rangle=\sum_{\xv,\xiv,z}U_{\xv,\xiv,z}e_{\xv,\xiv}(B)e_{\xv,\xiv}(C)|z\rangle\]

We will now modify $\Qs$ to obtain $\Qs'$ that will have the same success probability as $\Qs$.  $\Qs'$ chooses a random $D\in\Cs$, and simulates $\Qs$ with oracle access to $A'=A+D=B+C+D$, by forwarding the queries $\Qs$ makes to its own oracle $A$, and then performing a $D$ oracle query itself.  Since $\Qs'$ created $D$, making the $D$ oracle query does not cost $\Qs'$ any queries to $A$.  $\Qs'$ simply outputs the $z$ outputted by $\Qs$, as well as $D$.  $\Qs$ still sees a random oracle drawn from $\As$, so with probability $\epsilon$ its output $z$ will encode correct coset $C+D$.  From this and knowledge of $D$, we can recover $C$.  Therefore, $\Qs'$ still succeeds with probability $\epsilon$.

We can write the final state of this modified algorithm as

\[|\psi'_A\rangle=\frac{1}{\sqrt{|\Cs|}}\sum_{\xv,\xiv,z,D}U_{\xv,\xiv,z}e_{\xv,\xiv}(B)e_{\xv,\xiv}(C)e_{\xv,\xiv}(D)|z,D\rangle\]

The density matrix $\rho_C$ obtained by taking a random sample from $|\psi'_{B+C}\rangle$ as $B$ varies is then:

\[\rho_C=\frac{1}{|\Cs|}\sum_{B,\xv,\xiv,z,D,\xv',\xiv',z',D'}U_{\xv,\xiv,z}U^\dagger_{\xv',\xiv',z'}e_{\xv,\xiv}(B)e_{\xv,\xiv}(C+D)e_{\xv',\xiv'}(-B)e_{\xv',\xiv'}(-C-D')|z,D\rangle\langle z',D'|\]

Now we isolate $\sum_B e_{\xv,\xiv}(B)e_{\xv',\xiv'}(-B)=\sum_B(e_{\xv,\xiv}/e_{\xv',\xiv'})(B)$.  If $e_{\xv,\xiv}\neq e_{\xv',\xiv'}$ when restricted to $\Bs$, then $(e_{\xv,\xiv}/e_{\xv',\xiv'})$ is a character of $B$ that is not identically 1.  In this case, the sum goes to 0.  If $e_{\xv,\xiv}= e_{\xv',\xiv'}$ when restricted to $\Bs$, then the sum goes to $|\Bs|$.  This has the effect of forcing $e_{\xv,\xiv}\;\equiv_f\;e_{\xv',\xiv'}$ in the expression for $\rho_C$.  Therefore, we can write 

\[\rho_C=\sum_{e^*\in E_{\As,q,f}}p_{e^*}|\phi_{C,e^*}\rangle\langle\phi_{C,e^*}|\]
	
Where $\sum_{e^*\in E_{\As,q,f}}$ means that exactly one $e^*$ is chosen arbitrarily from each equivalence class on $E_{\As,q}$ induced by $\equiv_f$, and where
\begin{align*}
	p_{e^*}&=\frac{|\Bs|}{|\Cs|}\sum_{z,D}\left|\sum_{\xv,\xiv:\; e_{\xv,\xiv}\;\equiv_f\; e^*}U_{\xv,\xiv,z}e_{\xv,\xiv}(C+D)\right|^2=\frac{1}{|\Cs|}\sum_{z,D}\left|\sum_{\xv,\xiv:\;e_{\xv,\xiv}\;\equiv_f\;e^*}U_{\xv,\xiv,z}e_{\xv,\xiv}(D)\right|^2\\
	|\phi_{C,e^*}\rangle&=\frac{1}{\sqrt{p_{e^*}}}\sqrt{\frac{|\Bs|}{|\Cs|}}\sum_{z,D,\xv,\xiv:\;e_{\xv,\xiv}\;\equiv_f\;e^*}U_{\xv,\xiv,z}e_{\xv,\xiv}(C+D)|z,D\rangle\\
\end{align*}

Notice that $p_e^*$ does not depend on $C$.  Also, $\sum_{e^*} p_{e^*} = 1$ --- while it is not obvious given the expressions above that this should hold, it follows easily from the fact that the density matrix has trace 1 and the $|\phi_{C,e^*}\rangle$ are properly normalized.

\medskip

We can therefore think of the algorithm $\Qs'$ as acting in a different model where it gains access to $C$ in the following way: $\Qs'$ is simply given $e^*$\footnote{Actually, providing $e^*$ to $\Qs'$ gives $\Qs'$ potentially \emph{more} power in determining $C$.  This can only make the lower bound harder.  In particular, proving a lower bound for $\Qs'$ that is given $e^*$ implies a lower bound for $\Qs$ that is \emph{not} given $e^*$}, sampled with probability $p_{e^*}$, as well as the pure state $|\phi_{C,e^*}\rangle$, and its goal is to determine $C$.  We now bound the probability that it can do so.  

The Rank method of Boneh and Zhandry~\cite{EC:BonZha13} then gives us the following.  Conditioned on receiving $e^*$, which is independent of $C$, the maximum success probability is at most $\frac{1}{|\Cs|}$ times the dimension of the space spanned by the vectors $|\phi_{C,e^*}\rangle$ for various $C$.

Fix some $e^*$.  Notice that we can write $|\phi_{C,e^*}\rangle$ as:
\[|\phi_{C,e^*}\rangle=\sum_{e\in E_{\As,q,f,e^*}:\;e\;\equiv_f\;e^*}e(C)|\xi_e^{(e^*)}\rangle\]
Where \[
|\xi_e^{(e^*)}\rangle=\frac{1}{\sqrt{p_{e^*}}}\sqrt{\frac{|\Bs|}{|\Cs|}}\sum_{z,D}\left(\sum_{\xv,\xiv:\;e_{\xv,\xiv}=e}U_{\xv,\xiv,z}e(D)\right)|z,D\rangle\]

Therefore, for each $C$, $|\phi_{C,e^*}\rangle$ is linear combination of the basis functions $|\xi_{e}^{(e^*)}\rangle$ for $e\in E_{\As,q,f,e^*}$.  Therefore, the space spanned by the $|\phi_{C,e^*}\rangle$ as $C$ varies (but $e^*$ remains fixed) has dimension at most $|E_{\As,q,f,e^*}|$.  Therefore, the success probability, conditioned on receiving $e^*$, is at most $|E_{\As,q,f,e^*}|/|\Cs|$.  The overall success probability is at most the maximum of this quantity as we vary $e^*$, or 
\[\frac{1}{|\Cs|}\max_{e^*\in E_{\As,q}}|E_{\As,q,f,e^*}|=P_{\As,q,f}\]
This completes the proof.

\section{Applications}

\label{sec:apps}

\subsection{Parity}

The parity problem asks, given an oracle $A:[0,M-1]\rightarrow G$, to compute $M(G)=\sum_{x\in\Xs}A(x)$.

\begin{theorem}\label{thm:parity} The maximum success probability of any $q$-query quantum algorithm in the Parity problem is \[\min\left(\frac{\lfloor M/(M-q)\rfloor}{|G|},1\right)\]
\end{theorem}

We make the following observations:
\begin{itemize}
	\item In the case $|G|=2$, we obtain that the success probability is $\frac{1}{2}\lfloor M/(M-q)\rfloor$.  For $q<M/2$, this is equal to $1/2$, meaning the best quantum algorithm cannot beat random guessing.  Meanwhile, for $q\geq M/2$, the success probability becomes 1, meaning the parity can be computed with certainty.  This re-establishes the known bounds for the parity function due to Farhi et al.~\cite{FGGS98} and Beals et al.~\cite{BBCMdW01}.
	\item For more general $G$, the success probability is piecewise constant as $q$ varies, an increases by $1/|G|$ each time $q$ crosses $\frac{\ell-1}{\ell}M$ for integers $M$.
	\item To achieve success probability 1, $q\geq \lceil M(1-\frac{1}{|G|})\rceil$.  In particular, for $M<|G|$, $q=M$ queries are required to achieve a perfect algorithm.
	\item Setting $q=M-1$, the maximum success probability is $M/|G|$.  Thus for quantum algorithms to be able to save even a single query relative to classical algorithms in the bounded error setting, it must be that $M=\Omega(|G|)$.
	\item If error $p$ is tolerated, then approximately $q\approx (1-\frac{1}{p|G|})M$ queries are required.  For exponential $|G|$ and non-negligible $p$, all but a negligible fraction $M$ queries must be made.  Therefore, the advantage over classical algorithms is negligible.
\end{itemize}

Now we prove Theorem~\ref{thm:parity}.  Let $\Bs$ be the kernel of $f$: that is, the set of $A$ such that $\sum_{x\in\Xs}A(x)=0$.  The quotient group can be taken to be $\Cs$, the set of functions $A$ such that $A(x)=0$ for all $x\neq 0$.

Decompose $G$ as $G=\Z_{N_1}\times\cdots\times \Z_{N_k}$, and make $G$ a ring using component-wise multiplication.  The unit in this ring is the all-1s vector.  Notice that $\Bs$ and $\Cs$ are free modules over this ring structure.  For $\Bs$, use the basis $\{B_y\}_{y\in\Xs\setminus \{0\}}$ where \[B_y(x)=\begin{cases}
-1&\text{if }x=0\\
1&\text{if }x=y\\
0&\text{otherwise}
\end{cases}\]

For $\Cs$ we use the basis consisting of $C_0(x)=\begin{cases}1&\text{if }x=0\\0&\text{otherwise}\end{cases}$.  Let $\Bm(\xv)$ be the $(M-1)\times q$ matrix 
\[\Bm(\xv)=\left(\begin{array}{cccc}
B_1(x_1)     & B_1(x_2)     & \cdots & B_1(x_q)\\
B_2(x_1)     & B_2(x_2)     & \cdots & B_2(x_q)\\
\vdots       & \vdots       & \ddots & \vdots  \\
B_{M-1}(x_1) & B_{M-1}(x_2) & \cdots & B_{M-1}(x_q)
\end{array}\right)\]

Fix some $\hv$ in $G^{M-1}$.  We need to consider the set of $(\xv,\rv)$ such that $\Bm(\xv)\cdot\rv=\hv$.  Consider some $\xv$.  There are two cases:
\begin{itemize}
	\item $\xv$ does not contain 0.  In this case, the only rows of $\Bm(\xv)$ that are non-zero are those corresponding to the elements in $\xv$, of which there are $q$ (recall that we can assume the elements in $\xv$ are distinct).  Therefore, the column span of this matrix contains vectors that are zero everywhere outside positions corresponding to the elements of $\xv$.  In particular, if $\hv$ is in the image of $\Bm(\xv)$, then $\hv$ must be zero on at least $M-1-q$ points.
	\item $\xv$ contains 0.  Since we assume $\xv$ is ordered, $x_1=0$.  Then $\Bm(\xv)$ will be $-1$ everywhere in the first column.  The remaining columns will have exactly a single 1 in them.  Therefore, $\Bm(\xv)\cdot\rv$ will be the vector that is $-r_1$ everywhere, plus a vector that is non-zero in at most $q-1$ positions.  Thus if $\hv=\Bm(\xv)\cdot\rv$, it must be that $\hv$ is equal to $-r_1$ on at least $M-1-(q-1)=M-q$ points.
\end{itemize}

In either case, $\Bm(\xv)$ is full rank (since we assume the elements of $\xv$ are distinct).  Therefore, for any $\xv$, there is at most one $\rv$.

Next, we consider values of $z=C_0(\xv)\cdot \rv$ for values $(\xv,\rv)$ satisfying the above.  There are two cases:
\begin{itemize}
	\item $\xv$ does not contain 0.  In this case,  $C_0(\xv)$ is zero everywhere, so $z=0$, regardless of $\rv$
	\item $\xv$ contains 0 (so $x_1=0$).  Then $z=r_1$.  
\end{itemize}

Therefore, for every non-zero possible value of $z$, at least $M-q$ points of $\hv$ are equal to (the negative of) that value.  Moreover if $0$ is a possible value of $z$, $\hv$ is zero on at least $M-1-q$ points.  Therefore, if there are $k$ distinct possible values of $z$, it must be that \[M-1\geq (k-1)(M-q)+M-q-1\Longleftrightarrow k\leq\frac{M}{M-q}\]
Therefore $k=\lfloor M/(M-q)\rfloor$ is an upper bound on the number of possible $k$ values.  It is also straightforward to show that this bound is attainable, as long as the bound is at most 1.  Since $|\Cs|=|G|$, we have that the best probability of success of any $q$-query algorithm is:
\[\min\left(\frac{\lfloor M/(M-q)\rfloor}{|G|},1\right)\]

\subsection{Oracle Interrogation}

The oracle interrogation problem asks, given an oracle $A:\Xs\rightarrow G$, to output $k$ distinct input/output pairs given only $q$ quantum queries, where $q<k$.  Let $M=|\Xs|$ and $N=|G|$.

\begin{theorem}\label{thm:interrogate} The maximum success probability of any $q$ query quantum algorithm in the Oracle Interrogation problem is:
	\[\frac{1}{N^k}\sum_{i=0}^q\binom{M}{i}(N-1)^i\]
\end{theorem}

This reproves the bound given by Boneh and Zhandry~\cite{EC:BonZha13}.  To prove Theorem~\ref{thm:interrogate}, we rephrase Oracle Interrogation as a Quantum Oracle Classification problem.  The $k$ distinct inputs produced by the adversary specify the homomorphism that maps a function $A$ to its outputs on those $k$ points.  The $k$ outputs produced by the adversary are then the adversary's attempt to evaluate the homomorphism.  Using Lemma~\ref{lem:adversarial}, we know that it suffices fix a set of $k$ distinct inputs, say $S$, and consider an adversary that tries to evaluate $A$ on just those inputs.

$\Bs$, the kernel of the homomorphism defined by $S$, consists of functions that are zero on $S$, and quotient group $\Cs$ can be taken to consist of functions that are zero everywhere except $S$.  

As with parity, we can assign a ring structure to $G$, and $\Bs,\Cs$ will be free modules over $G$.  Let $B_y(x)=\begin{cases}
1&\text{if }x=y\\0&\text{otherwise}\end{cases}$.  Then $\Bs$ is spanned by $B_y$ for $y\notin S$, and $\Cs$ is spanned by $B_y$ for $y\in S$.  Let $\Bm(\xv),\Cm(\xv)$ be as in Corollary~\ref{cor:main}.  The possible values of $\zv=\Cm(\xv)\cdot\rv$ are just the vectors of length $k$ that are non-zero on at most $q$ points.  Of course, we need to choose some $\hv$ and restrict to $(\xv,\rv)$ such that $\hv=\Bm(\xv)\cdot\rv$.  It is straightforward to see that setting $\hv=0$ will still allow all vectors of at most $q$ points even after the restriction.  

Thus, the total number of $\zv$ is \[\sum_{i=0}^q\binom{M}{i}(N-1)^i\]
So the maximum success probability is \[\frac{1}{N^k}\sum_{i=0}^q\binom{M}{i}(N-1)^i\]

\subsection{Polynomial Interpolation}

The polynomial interpolation problem asks, given a random degree-$d$ polynomial $A$ over a field $\F$, to determine the $d+1$ coefficients of $A$.  We prove the following theorem:

\begin{theorem}\label{thm:interpolate}The maximum success probability of a $q$-query algorithm in the polynomial interpolation problem is
	\begin{enumerate}
		\item\label{line:small} At most $\frac{1}{|\F|}$ if $q\leq \frac{d}{2}$
		\item\label{line:med} Between $\frac{1}{q!}\left(1-\binom{q+1}{2}\frac{1}{|\F|}\right)$ and $\frac{1}{q!}$ if $q=\frac{d}{2}+\frac{1}{2}$
		\item\label{line:big} At least $1-\frac{(d/2+1)!+d+1}{|\F|}$ if $q\geq \frac{d}{2}+1$
	\end{enumerate}
\end{theorem}

In particular, for $|\F|\gg d!$, we have that the maximum success probability is vanishingly small for $q\leq \frac{d}{2}$, is roughly $\frac{1}{q!}$ for $q=\frac{d}{2}+\frac{1}{2}$, and is overwhelming for $\frac{d}{2}+1$.  Thus, for even $d$, there is a sharp threshold at $d/2$, exactly half the threshold of the classical setting.  For odd $d$, however, the story depends on how large $d$ and $q$ are.  If $d$ and $q$ are taken to be logarithmic or larger, so that $1/q!$ is negligible, then there is a sharp threshold as $q$ goes from $d/2+1/2$ to $d/2+3/2$.  However, if $d,q$ are constants, then there is no sharp threshold at all: the success probability goes from vanishing at $d/2-1/2$ to constant at $d/2+1/2$ to overwhelming at $d/2+3/2$.  

\medskip

Now we prove Theorem~\ref{thm:interpolate}.  The homomorphism for Polynomial Interpolation is just the identity homomorphism.  Thus the kernel $\Bs$ is just the zero function, and $\Cs$ is the set of all degree-$d$ polynomials.  $\Cs$ is spanned by the functions $C_i(x)=x^i$ for $i=0,\dots,d$.  Therefore, the matrix $\Cm(\xv)$ in Corollary~\ref{cor:main} is just the Vandermonde matrix
\[\Cm(\xv)=\left(\begin{array}{cccc}
1& 1& \cdots & 1\\
x_1& x_2& \cdots & x_q\\
\vdots & \vdots & \ddots & \vdots \\
x_1^d & x_2^d& \cdots & x_q^d
\end{array}\right)\]

Meanwhile, the $\Bm$ matrix is empty, meaning there is no constraint on the $\xv,\rv$ pairs allowed.  Our goal, therefore, is to count the number of vectors $\zv$ that can be represented as $\zv=\Cm(\xv)\cdot\rv$ for $\xv$ that are subsets of $|\F|$ of size $q$, and $\rv\in|\F|^q$.  Let $N_{d,q}$ be the number of such distinct vectors.  Then the total success probability by Corollary~\ref{cor:main} is $N_{d,q}/|\F|^{d+1}$.  

We can trivially bound $N_{d,q}$ from above by $\binom{|\F|}{q}|\F|^q\leq |\F|^{2q}/q!$.  This gives us the upper bounds in parts~\ref{line:small} and~\ref{line:med}.  

\smallskip

For the lower bound in part~\ref{line:med}, we want to show that the number of vectors is close to the number of pairs $(\xv,\rv)$.  Suppose that two distinct pairs $(\xv,\rv)$ and $(\xv',\rv')$ map to the same vector: $\Cm(\xv)\cdot \rv=\Cm(\xv')\cdot\rv'$.  We can write this equation as:
\begin{equation}\label{eq:lower}
\zerom=\Mm\cdot\sv=\left(\begin{array}{cccc|cccc}
1      & 1      & \cdots & 1      & 1      & 1      & \cdots & 1      \\
x_1    & x_2    & \cdots & x_q    & x_1'   & x_2'   & \cdots & x_q'   \\
\vdots & \vdots & \ddots & \vdots & \vdots & \vdots & \ddots & \vdots \\
x_1^d  & x_2^d  & \cdots & x_q^d  & x_1'^d & x_2'^d & \cdots & x_q'^d
\end{array}\right)\cdot
\left(\begin{array}{c}r_1\\r_2\\\vdots\\r_q\\\hline -r_1'\\-r_2'\\\vdots\\-r_q'\end{array}\right)\end{equation}

Since we are considering the case where $d+1=2q$, the matrix $\Mm$ is square.  Moreover, $\Mm$ is also a Vandermonde matrix, and it is full rank if $\xv$ and $\xv'$ do not overlap.  In this case, the only solution is $\sv=\zerom$, which means $\rv=\rv'=\zerom$.  If there is an overlap between $\xv,\xv'$, reorder $\xv,\xv'$ so that the first $k$ elements are identical, and the remaining $q-k$ are distinct, and apply the same reordering to $\rv,\rv'$.  We can then rewrite Equation~\ref{eq:lower} as \begin{equation}
\label{eq:lower2}\zerom=\Mm'\cdot\sv' = \left(\begin{array}{ccc|ccc|ccc}
1      & \cdots & 1      & 1         & \cdots & 1      & 1          & \cdots & 1      \\
x_1    & \cdots & x_k    & x_{k+1}   & \cdots & x_q    & x_{k+1}'   & \cdots & x_q'   \\
\vdots & \ddots & \vdots & \vdots    & \ddots & \vdots & \vdots     & \ddots & \vdots \\
x_1^d  & \cdots & x_k^d  & x_{k+1}^d & \cdots & x_q^d  & x_{k+1}'^d & \cdots & x_q'^d
\end{array}\right)\cdot \left(\begin{array}{c}r_1-r1'\\\vdots\\r_k-r_k'\\\hline r_{k+1}\\\vdots\\r_q\\\hline r_{k+1}'\\\vdots\\r_q'\end{array}\right)\end{equation}

Now $\Mm'$ is full rank, and it is ``tall''  since its width is $2q-k<2q$ and its height is $d+1=2q$.  Therefore, the only solution is $\sv'=0$, which means $r_i=r_i'$ for $i=1,\dots,k$ and $r_i=r_i'=0$ for $i=k+1,\dots,q$.  Notice that since $\rv=\rv'$, it must be that $\xv\neq\xv'$, and therefore $k<q$.  Therefore, $r_q=r_q'=0$.  

Thus, if we restrict to $\rv$ that are non-zero in every coordinate, there are \emph{no} collisions among the $(\xv,\rv)$.  Counting the number of pairs of vectors $(\xv,\rv)$ where $r_i\neq 0\forall i$ then gives us a lower bound on $N_{d,q}$.  The number of such vectors is $\binom{|\F|}{q}(|\F|-1)^q\geq \frac{|\F|^{2q}}{q!}\left(1-\binom{q+1}{2}\frac{1}{|\F|}\right)$.  Dividing by $|\F|^{d+1}=|\F|^{2q}$ gives us the lower bound in part~\ref{line:med} of the theorem.

\smallskip

Finally, it remains to prove the lower bound in part~\ref{line:big}.  Again, the goal is to show that there are not \emph{too} many collisions in the function $(\xv,\rv)\mapsto\zv=\Cm(\xv)\cdot\rv$.  However, this time there \emph{will} be many collisions, as the number of pairs is $\approx |\F|^{2q}/q!$, whereas the co-domain has size $|\F|^{d+1}\geq |\F|^{2q-1}\ll |\F|^{2q}/q!$.  

We will prove a lower bound on the number of vectors by proving an upper bound on the collision probability.  That is, we will choose a distribution $D$ on pairs $(\xv,\rv)$, and show that collisions happen with probability at most $p$.  It then follows that the image size must be at least $1/p$.

We prove the lower bound in the case $q=d/2+1$; the lower bound easily extends to larger $q$.  Let $D$ be the distribution that picks $\xv$ uniformly at random from the subsets of $\F$ of size $q$, and $\rv$ uniformly at random from $(\F\setminus\{0\})^q$.  Consider sampling two pairs $(\xv,\rv),(\xv',\rv')$.  We will now bound the probability they collide.  We consider three cases:
\begin{itemize}
	\item $\xv$ and $\xv'$ do not overlap.  Look again at the matrix $\Mm$ from Equation~\ref{eq:lower}.  This matrix is no longer square, but is actually wide, since its width is $2q$, but its height is $d+1<2q$.  Therefore, its rank is $d+1$, and its null space has dimension $2q-(d+1)=1$.  This means that $\sv=(\rv,-\rv')$, which is uniform on $(\F\setminus\{0\})^{2q}$ has probability at most $\frac{|\F|-1}{(|\F|-1)^{2q}}=\frac{1}{(|\F|-1)^{d+1}}$ of being in the null space.
	\item $\xv$ and $\xv'$ overlap but are not equal.  In this case, we look at the matrix $\Mm'$ in Equation~\ref{eq:lower2}.  It is full rank, and its width is now only $2q-k\leq d+1$.  Therefore as before, the only solutions have $r_q=r_q'=0$.  Since we are restricting to $r_q\neq 0$, there are \emph{no} collisions in this case.
	\item $\xv=\xv'$.  In this case, we have a collision if and only if $\rv=\rv'$, which happens with probability $\frac{1}{(|\F|-1)^{2q}}$.  Notice that $\xv=\xv'$ with probability $1/\binom{|\F|}{q}$
\end{itemize}

Therefore the total collision probability is bounded by \[\frac{1}{(|\F|-1)^{d+1}}+\frac{1}{\binom{|\F|}{q}(|\F|-1)^q}\]

The image size of the map $(\xv,\rv)\mapsto\Cm(\xv)\cdot\rv$ is at least the reciprocal of this probability.  Through a straightforward though tedious series expansion for large $|\F|$, the reciprocal can be bounded from below by

\[|\F|^{d+1}\left(1-\frac{q!+2q-1}{|\F|}\right)\]

Dividing by $|\F|^{d+1}$ and setting $q=d/2+1$ gives the desired result.

\medskip

Next, we consider two variants of polynomial interpolation.

\subsection{Polynomial Evaluation}

In the polynomial evaluation problem, we are given oracle access to a random degree-$d$ polynomial $A$ over a field $\F$, and we are tasked with evaluating the polynomial are certain inputs. Since we can easily learn as many points as we made queries, our goal is to learn \emph{more} points than the number of queries made.  Specifically, given $q$ queries, we are asked to produce $q+1$ distinct input/output pairs.

\begin{theorem}\label{thm:eval} For $q\leq d/2$, the maximum success probability of any $q$-query quantum algorithm in the Polynomial Evaluation problem is at most $(q+1)qe^{2\sqrt{q}}/|\F|$.
\end{theorem}

This shows that setting $d=2q$ suffices for a degree-$d$ polynomial to be a $q$-time quantum-secure message authentication code (MAC).  This improves on the previously best bound of $d=3q$ due to Boneh and Zhandry~\cite{EC:BonZha13}.  In particular, for a 1-time MAC, degree 2 is sufficient.  As it is known that degree 1 is insufficient, this resolves the question for 1-time MACs.  More generally, for any constant $q$, $q$ queries suffice to recover a degree $2q-1$ polynomial with constant probability.  Therefore, degree $d=2q$ is required.  In the case where $q$ is allowed to grow, a more careful analysis is required to see if degree $2q$ is optimal.

\medskip

Now we prove Theorem~\ref{thm:eval}.  Fix any $q+1$ distinct inputs $\Ys$ that the adversary will try to evaluate $A$ on.  The null space $\Bs$ is the set of polynomials that vanish on $\Ys$.  Such polynomials can be written as $A(x)=Q(x)A'(x)$ where $Q(x)=\prod_{y\in\Ys}(x-y)$ and $A'(x)$ is a random degree $d-q-1$ polynomial.  This space is spanned by $Q(x)x^i$ for $i=1,\dots,d-q-1$.  Therefore, the condition $\hv=\Bm(\xv)\cdot\rv$ can be written as:
\begin{align}
\hv&=\left(\begin{array}{cccc}
Q(x_1)            & Q(x_2)            & \cdots & Q(x_q)            \\
Q(x_1)x_1         & Q(x_2)x_2         & \cdots & Q(x_q)x_q         \\
\vdots            & \vdots            & \ddots & \vdots            \\
Q(x_1)x_1^{d-q-1} & Q(x_2)x_2^{d-q-1} & \cdots & Q(x_q)x_q^{d-q-1}
\end{array}\right)\cdot \left(\begin{array}{c}r_1\\r_2\\\vdots\\r_q\end{array}\right)\notag\\
&=
\left(\begin{array}{cccc}
1           & 1           & \cdots & 1           \\
x_1         & x_2         & \cdots & x_q         \\
\vdots      & \vdots      & \ddots & \vdots      \\
x_1^{d-q-1} & x_2^{d-q-1} & \cdots & x_q^{d-q-1}
\end{array}\right)\cdot \left(\begin{array}{c}Q(x_1)r_1\\Q(x_2)r_2\\\vdots\\Q(r_q)r_q\end{array}\right)=:\Mm'(\xv)\cdot\rv'\label{eq:eval}
\end{align}

$\Mm'(\xv)$ is the usual Vandermonde matrix, which is full rank.  Since $d\geq 2q$, the matrix is at least as tall as it is wide, and therefore, for a given $\hv,\xv$, there is exactly one solution for the vector $\rv'$.  As $r_i'=Q(x_i)r_i$, this uniquely determines each of the $r_i$, except in the case that $r'_i=0$ and $x_i\in\Ys$ (so that $Q(x_i)=0$).  If $k$ of the $x_i$ are not in $\Ys$ (so $q-k$ are in $\Ys$), we see that the maximum number of solutions is $|\F|^{q-k}$.  The number of $\xv$ that overlap with $\Ys$ in $q-k$ points is $\binom{q+1}{q-k}\binom{|\F|-q-1}{k}$.  Therefore, for any given $\hv$, the total number of $(\xv,\rv)$ pairs satisfying Equation~\ref{eq:eval} is therefore at most 
\[\sum_{k=0}^q \binom{q+1}{q-k}\binom{|\F|-q-1}{k} |\F|^{q-k}\leq \sum_{k=0}^q\binom{q+1}{q-k} \frac{|\F|^{k}}{k!}|\F|^{q-k}\leq |\F|^q (q+1)\sum_{k=0}^q \binom{q}{k}\frac{1}{k!}=:|\F|^q N_q\]

Hence this is also a bound on the number of vectors $\zv=\Cm(\xv)\cdot\rv$ such that $\Bm(\xv)\cdot \rv=\hv$.  We now bound $N_q$.  Let $k=\alpha q$.  First, we recall that $\binom{q}{\alpha q}\leq e^{H_\alpha q}$ where $H_\alpha=-(1-\alpha)\ln(1-\alpha)-\alpha\ln \alpha$ is the entropy of a coin flip that is heads with probability $\alpha$.  Then we use Stirling's approximation to lower bound $k!\geq e^{k(\ln k-1)}$.  Therefore, $\ln \left(\binom{q}{\alpha q}/k!\right)\leq H_\alpha q -\alpha q(\ln (\alpha q)-1)=:P_q(\alpha)$.  It is straightforward to show that $P_q$ attains a maximum value at $\alpha=\frac{-1+\sqrt{4q+1}}{2a}$, and we can bound $P_q$ at this value by $2\sqrt{q}$.  Thus we can bound the summand in the expression for $q$ as $e^{2\sqrt{q}}$, giving a total bound of $N_q\leq (q+1)qe^{2\sqrt{q}}$.  

Hence the number of vectors $\Cm(\xv)\cdot\rv$ for any given $\hv$ is at most $(q+1)qe^{2\sqrt{q}}|\F|^q$.  Dividing by $|\F|^{q+1}$ (the total number of possible outputs on $q+1$ points), we obtain a maximum success probability of $\frac{(q+1)qe^{2\sqrt{q}}}{|\F|}$, as desired.

\subsection{Polynomial Extrapolation}

In the polynomial extrapolation problem, we are given oracle access to a random degree-$d$ polynomial $A$ over a field $\F$ everywhere except a single point, say 0, and we are tasked with evaluating the polynomial at 0.  This problem relates to a generalization of Shamir secret sharing where shares are superpositions.

\begin{theorem}\label{thm:extrapolate} The maximum success probability of a $q$-query quantum algorithm in the Polynomial Extrapolation problem is:
	\begin{enumerate}
		\item \label{line:small2} $1/|\F|$ if $q\leq d/2$.
		\item \label{line:med2} At most $\lfloor(|\F|-1)/q\rfloor/|\F|\approx 1/q$ for $q=(d+1)/2$.  Moreover, if $q|(|\F|-1)$, this is tight.
		\item \label{line:big2} At least $1-\frac{(d/2+1)!+d+1}{|\F|}$ if $q\geq (d+2)/2$.
	\end{enumerate}
\end{theorem}

This theorem shows that, for the application to quantum secret sharing discussed in the introduction, even $d$ give a tight threshold between the secret being hidden and being able to compute the secret with overwhelming probability.  However, for odd degree, there is no tight threshold, as long as $(d+1)/2$ divides $|\F|-1$.  This holds \emph{even if $q$ is polynomial}, in contrast to Polynomial Interpolation, where large polynomial $q$ give a tight threshold even for odd $d$.  We conjecture that our bounds for $q=(d+1)/2$ are tight for \emph{all} field sizes, in which case there is no tight threshold for any field size.

\medskip

We now prove Theorem~\ref{thm:extrapolate}.  Here, the kernel $\Bs$ is the set of degree-$d$ polynomials $A$ such that $A(0)=0$ (that is, the set of polynomials with no constant term).  Meanwhile, the quotient group $\Cs$ can be taken to be the set of constant polynomials.  Both of these groups form vector spaces: $\Bs$ is spanned by the polynomials $x,\dots,x^d$, and $\Cs$ is spanned by the constant polynomial $1$.  Therefore, the matrix $\Bm(\xv)$ in Corollary~\ref{cor:main} for $q$ queries is just the Vandermonde-style matrix
\[\Bm(\xv)=\left(\begin{array}{cccc}
x_1    & x_2    & \cdots & x_q    \\
x_1^2  & x_2^2  & \cdots & x_q^2  \\
\vdots & \vdots & \ddots & \vdots \\
x_1^d  & x_2^d  & \cdots & x_q^d
\end{array}\right)\]

Meanwhile the matrix $\Cm(\xv)$ is just $\Cm=(1\;1\;\cdots\; 1)$.  Suppose $d\geq q$, in which case $\Bm(\xv)$ is full rank for any set $\xv$.  Pick an $\hv\in\F^d$, and consider the set of solutions to $\hv=\Bm(\xv)\cdot\rv$.  First, since $\Bm(\xv)$ is full rank, there is at most a single $\rv$ for any $\xv$.  Now, consider an $\xv,\rv$ that is a solution. We want to determine the possible values of $z=\Cm\cdot \rv$, which is just the sum of the elements in $\rv$. 

Let $P_\xv(y)=y^q+\sum_{i=0}^{q-1}a_i y^i$ be the monic polynomial that is zero on the $q$ points of $\xv$.  The $a_i$, as functions of the $x_j\in\xv$, are simply the elementary symmetric polynomials in $q$ variables, up to a sign factor.  In particular, $a_0$ is (up to sign) the product of the $x_j$, and therefore $a_0\neq 0$ (since $\xv$ is restricted to be non-zero).  Moreover, for a given $x\in\xv$, $a_i$ can be written as $a_i=b_i x+b_{i-1}$ where the $b_i$ are the elementary symmetric polynomials (up to signs) in $q-1$ variables, evaluated on the $x_j\in\xv\setminus\{x\}$.  This uses the convention that $b_{q-1}=1,b_{-1}=0$.  

Define $\hat{a}_i=\frac{a_i}{a_0}$ for $i\in[q-1]$, and $\hat{a_q}=\frac{1}{a_0}$.  Then define $\hat{P}_{\xv}(y):= 1+\sum_{i=1}^q \hat{a}_i y^i=\frac{1}{a_0}p_{\xv}(y)$.  Therefore, $\hat{P}_{\xv}$ also has zeros on exactly the points in $\xv$.  Define $\hat{\av}$ to be the column vector of the $\hat{a}_i$.  Notice that $\hat{\av}$ uniquely determines the polynomial $\hat{P}$, which uniquely determines $\xv$, if it exists.  However, $\hat{\av}$ may determine a polynomial $\hat{P}$ which does not split over $\F$, or which has multiple roots, in which case $\xv$ does not exist.  Similarly define $\hat{b}=\frac{b_i}{b_0}$.  Note that $\hat{a}_i=\hat{b}_i+\frac{1}{x}\hat{b}_{i-1}$.  Define $\hat{\bv}$ to be the column vector of the $\hat{b}_i$

First, we observe that $-1=\sum_{i=1}^q \hat{a}_i x_j^i$ for each $x_j\in\xv$.  Define $\Bm_{[1,q]}(\xv)$ to consist of the first $q$ rows of $\Bm(\xv)$.  Then we observe that $\hat{\av}^T\cdot\Bm_{[1,q]}(\xv)=(-1\;-1\;\cdots\;-1)=-\Cm$.  Then we can write \[z=\Cm\cdot\rv=-\hat{\av}^T\cdot\Bm_{[1,q]}(\xv)\cdot\rv=-\hat{\av}^T\cdot \hv_{[1,q]}\]
where $\hv_{[1,q]}=(h_1\;h_2\;\dots\; h_q)$ is the first $q$ elements of $\hv$.  Therefore, we can determine $z$ simply from $\xv$ and $\hv$ without computing $\rv$.  If we fix an $x\in\xv$, it is straightforward to re-write this as \[z=-\frac{h_1}{x}-\hat{\bv}^T\cdot \hv_{[1,q-1]}(x)\]
Where $h_i(x)=h_i+\frac{1}{x}h_{i+1}$ and $\hv_{[1,q-1]}(x)=(h_1(x)\;h_2(x)\;\cdots\; h_{q-1}(x))$.

Next, the fact that $\hv=\Bm(\xv)\cdot\rv$ means means $\hv$ is in the column space of $\Bm(\xv)$.  Define $\Dm(\xv)$ to be the matrix

\[\Dm(\xv)=\left(\begin{array}{ccccccccc}
1            & \hat{a}_1    & \cdots       & \hat{a}_{q-1} & \hat{a}_q     & 0         & \cdots        & 0        \\
0            & 1            & \hat{a}_1    & \cdots        & \hat{a}_{q-1} & \hat{a}_q & 0             & \cdots   \\
\vdots       & \ddots       & \ddots       & \ddots        & \ddots        & \ddots    & \ddots        & \vdots   \\
0            & \cdots       & 0            & 1             & \hat{a}_1     & \cdots    & \hat{a}_{q-1} & \hat{a}_q 
\end{array}\right)\]

Notice that $\Dm(\xv)\cdot \Bm(\xv)=\zerom$ since $\hat{P}(x_i)=0$ for each $x_i\in\xv$.  This means $\Dm(\xv)\cdot\hv=\zerom$. The resulting set of equations can be re-written as
\begin{equation}\label{eq:hs}\Hm\cdot\av=\left(\begin{array}{cccc}
h_2       & h_3       & \cdots & h_{q+1} \\
h_3       & h_4       & \cdots & h_{q+2} \\
\vdots    & \vdots    & \ddots & \vdots  \\
h_{d-q+1} & h_{d-q+2} & \cdots & h_d
\end{array}\right)\cdot \left(\begin{array}{c}\hat{a}_1\\\hat{a}_2\\\vdots\\\hat{a}_q
\end{array}\right)=-\left(\begin{array}{c}h_1\\h_2\\\vdots\\h_{d-q}
\end{array}\right)=-\hv_{[1,d-q]}\end{equation}

If we fix an $x\in\xv$, it is straightforward to re-write this as
\begin{equation}\label{eq:hs2}\Hm(x)\cdot\bv=\left(\begin{array}{cccc}
h_2(x)       & h_3(x)       & \cdots & h_{q}(x) \\
h_3(x)       & h_4(x)       & \cdots & h_{q+1}(x) \\
\vdots    & \vdots    & \ddots & \vdots  \\
h_{d-q+1}(x) & h_{d-q+1}(x) & \cdots & h_{d-1}(x)
\end{array}\right)\cdot \left(\begin{array}{c}\hat{b}_1\\\hat{b}_2\\\vdots\\\hat{b}_{q-1}
\end{array}\right)=-\left(\begin{array}{c}h_1(x)\\h_2(x)\\\vdots\\h_{d-q}(x)
\end{array}\right)=-\hv_{[1,d-q]}(x)\end{equation}

\paragraph{The case $d=2q$.}  In this case, $\Hm$ is a square $q\times q$ symmetric matrix.  Fix $\hv$, which determines $\Hm$.  If $\Hm$ happened to be full rank, then there would be a unique solution $\hat{\av}$ to the equation above, meaning a unique set $\xv$.  In this case there is only a single pair $(\xv,\rv)$, meaning only a single $z$.  We will now show that this is true even if $\Hm$ is not full rank.  Pick an arbitrary solution $\hat{\av}_0$ to $\Hm\cdot\hat{\av}_0=-\hv_{[1,d-q]}=-\hv_{[1,q]}$.  Then \emph{any} solution can be written as $\hat{\av}_0+\hat{\av}$ for some $\hat{\av}$, where $\Hm\cdot\hat{\av}=0$.  Then consider the possible $z$ values, which are $-\langle\hv_q,\hat{\av}_0+\hat{\av}\rangle$.  This can be expressed as
\begin{align*}
z&=(\hat{\av}_0+\hat{\av})^T\cdot (-\hv_{[1,q]}) = (\hat{\av}_0+\hat{\av})^T\cdot \Hm\cdot (\hat{\av}_0+\hat{\av})\\
 &=\hat{\av}_0^T\cdot\Hm\cdot\hat{\av}_0+\hat{\av}^T\cdot\Hm\cdot\hat{\av}_0+\hat{\av}_0\cdot\Hm\cdot\hat{\av}+\hat{\av}\cdot\Hm\cdot\hat{\av}\\
 &=\hat{\av}_0^T\cdot\Hm\cdot\hat{\av}_0+0+0+0\\
\end{align*}

Thus, for any $\Hm$ (and hence for any $\hv$), there is exactly one value of $z$.  Applying Corollary~\ref{cor:main}, we see that the maximal success probability is $1/|\F|$, the same as random guessing.

\paragraph{The case $d=2q-1$.} In this case, $\Hm$ is no longer square, but instead has dimensions $(q-1)\times q$.  First, fix an $x\in\F$ to be included in $\xv$.  We now wish to count, for any fixed $\hv$, the total number of values $z$ takes on as we vary the other values in $\xv$.  Recall that we can write $z$ as $-\frac{h_1}{x}-\hat{\bv}^T\cdot \hv_{[1,q-1]}(x)$ where $\hat{\bv}$ is a solution to $\Hm(x)\cdot \hat{\bv} = -\hv_{[1,d-q]}(x)$ in Equation~\ref{eq:hs2}.  Since $x$ is fixed, it suffices to count the number of values $-\hat{\bv}^T\cdot\hv_{[1,q-1]}(x)$ takes on.  This counting problem is exactly a ``smaller'' instance of the Polynomial Extrapolation problem with $q'=q-1$ queries and degree $d'=d-1$ polynomials.  Notice that $d'=2q'$, so this is the problem we already solved above.  In particular, there is at most one $z$ value for this fixed $x$.  

Therefore, every non-zero $x\in\F$ is associated with at most a single $z$ value.  Moreover, for each $z$, there are at least $q$ such non-zero $x$ (since each $z$ is derived from some $\xv$ of $q$ elements).  Therefore, there are at most $(|\F|-1)/q$ possible values of $z$.  This completes the upper bound. 

\medskip

For the lower bound, we consider the case where $q|(|\F|-1)$.  This means $q$ divides the order of the multiplicative group of non-zero elements of $\F$.  This means there is a primitive $q$th root of unity $\omega$, $\omega^q=1$.  We will let $h=(0\;\cdots\; 0\;q\;0\;\cdots\;0)$ be the vector of $2q-2$ 0's with a $q$ in the middle.  Then for any $x\in\F$, consider the vector $\xv=(x\;\omega x\; \dots\;\omega^{q-1}x)$ and $\rv=(1\;1\;\dots\;1)/x^q$.  Then the $j$th row of $\Bm(\xv)$ is $x^j\cdot(1\;\omega^j\;\dots\;\omega^{j(q-1)})$.  The $j$th element of $\Bm(\xv)\cdot\rv$ is then \[\sum_{i=0}^{q-1} (x\omega^i)^j/x^q=x^{j-q}\sum_{i=0}^{q-1}\omega^{i j}\]

Notice that the sum is 0 unless $j\mod q=0$.  As $1\leq j\leq 2q-1$, the only value for which the sum is non-zero is $j=q$.  Moreover, for $j=q$, the sum goes to $q$, and thus the whole $q$th element evaluates to $q$.  Thus, $\Bm(\xv)\cdot\rv=\hv$, as desired. 

Now we need to count the number of values $z=\Cm(\xv)\cdot\rv$.  Recall that $\Cm(\xv)$ is just the all-ones vector, so we get that $\Cm(\xv)\cdot\rv=q/x^q$.  Since $q|(|\F|-1)$, we know that $q$ is relatively prime to the characteristic of $|\F|$, so $q$ is non-zero.  Moreover, $x^q$, when restricted to non-zero $x$, is exactly $q$-to-1.  Therefore, there are $(|\F|-1)/q$ different values for $z=q/x^q$.  This completes the lower bound.

\paragraph{The case $d<2q-1$.}  It is straightforward to re-work the analysis from the Polynomial Interpolation case to show that we can still interpolate the polynomial completely in this setting, even though we can no longer query on 0.  In fact, the exact same bound given in the Polynomial Interpolation setting (namely a success probability of at least $1-\frac{(d/2+1)!+d+1}{|\F|}$) can be given.

\bibliographystyle{alpha}
\bibliography{abbrev0,crypto,bib}

\appendix

\end{document}